\documentclass{llncs}

\usepackage{xspace,amsmath,url}

\usepackage{graphics}
\usepackage{epstopdf}
\usepackage{epsfig}

\newcommand{\ceil}[1]{\lceil #1 \rceil}

\newcommand{\exclude}[1]{}

\begin{document}

\title{New tabulation and sparse dynamic programming based techniques 
for sequence similarity problems}

\author{Szymon Grabowski}

\institute{
	  Lodz University of Technology, Institute of Applied Computer Science,\\
	  Al.\ Politechniki 11, 90--924 {\L}\'od\'z, Poland
	  \email{sgrabow@kis.p.lodz.pl}
}

\maketitle

\begin{abstract}
Calculating the length of a longest common subsequence (LCS) 
of two strings $A$ and $B$ of length $n$ and $m$
is a classic research topic, with many 
worst-case oriented results known.
We present two algorithms for LCS length calculation with 
respectively $O(mn \log\log n / \log^2 n)$ 
and $O(mn / \log^2 n + r)$ time complexity,
the latter working for $r = o(mn / (\log n \log\log n))$, 
where $r$ is the number of matches in the dynamic programming matrix.
We also describe conditions for a given problem 
sufficient to apply our techniques, with several concrete 
examples presented, namely the edit distance, LCTS 
and MerLCS problems.
\end{abstract}

%%%%%%%%%%%%%%%%%%%%%%%%%%%%%%%%%%%%%%%%%%%%%%%%
\section{Introduction}
%%%%%%%%%%%%%%%%%%%%%%%%%%%%%%%%%%%%%%%%%%%%%%%%
\noindent 
Measuring the similarity of sequences is an old research topic and 
many actual measures are known in the string matching literature.
One classic example concerns the computation of a  
longest common subsequence (LCS)
%% ~\cite{BFC2008,Gus1997}, 
in which a subsequence that is common to all sequences and has the 
maximal possible length is looked for.
A simple dynamic programming (DP) solution works in $O(mn)$ time for two 
sequences of length $n$ and $m$, respectively, 
but faster algorithms are known.
The LCS problem has many applications in diverse areas, 
like version control systems, 
%% comparison of DNA strings~\cite{ParvinniaTZ08}, 
comparison of DNA strings,
structural alignment of RNA sequences.
%% ~\cite{BeregKWZ07}. 
Other related problems comprise calculating the edit (Levenshtein) 
distance between two sequences, the longest common transposition-invariant 
subsequence, or LCS with constraints in which the longest common subsequence
of two sequences must contain, or exclude, some other sequence.

Let us focus first on the LCS problem, for two sequences $A$ and $B$. 
It is defined as follows. 
Given two sequences, $A = a_1 \ldots a_n$ and $B = b_1 \ldots b_m$, 
over an alphabet $\Sigma$ of size $\sigma$, find a longest subsequence 
$\langle a_{i_1}, a_{i_2}, \ldots, a_{i_{\ell}}\rangle$ of $A$ 
such that 
$a_{i_1} = b_{j_1}, a_{i_2} = b_{j_2},\ldots, a_{i_{\ell}} = b_{j_{\ell}}$, 
where $1 \leq i_1 < i_2 < \ldots < i_{\ell} \leq n$ 
and $1 \leq j_1 < j_2 < \ldots < j_{\ell} \leq m$.
The found sequence may not be unique.
W.l.o.g. we assume $n \geq m$.
To avoid uninteresting complications, we also assume that $m = \Omega(\log^2 n)$.
Additionally, we assume that $\sigma = O(m)$.
The case of a general alphabet, however, can be handled with standard means, 
i.e., we can initially map the sequences $A$ and $B$ onto 
an alphabet of size $\sigma' = O(m)$, in $O(n\log\sigma')$ time, 
using a balanced binary search tree.
We do not comprise this tentative preprocessing step 
in further complexity considerations.

Often, a simplified version of the LCS problem is considered, when 
one is interested in telling only the length of a longest common subsequence (LLCS).

In this paper we present two techniques for finding the LCS length, 
one (Section~\ref{sec:alg1}) based on tabulation and improving the result of 
Bille and Farach-Colton~\cite{BFC2008} by factor $\log\log n$, 
the other (Section~\ref{sec:alg2}) combining tabulation and 
sparse dynamic programming and being slightly faster if the number 
of matches is appropriately limited.
In Section~\ref{sec:apps} 
we show the conditions necessary to apply these algorithmic techniques. 
Some other, LCS-related, problems fulfill these conditions, so we immediately 
obtain new results for these problems as well.

Throughout the paper, we assume the word-RAM model of computation.
All used logarithms are base 2.

%%%%%%%%%%%%%%%%%%%%%%%%%%%%%%%%%%%%%%%%%%%%%%%%%%%%%%%%%%%%%%%%%%%%%%%%%%%%%%
\section{Related work}
\label{sec:related_work}
%%%%%%%%%%%%%%%%%%%%%%%%%%%%%%%%%%%%%%%%%%%%%%%%%%%%%%%%%%%%%%%%%%%%%%%%%%%%%%
\noindent 
A standard solution to the LCS problem is based on dynamic programming, 
and it is to fill a matrix $M$ of size $(n+1) \times (m+1)$, where each 
cell value depends on a pair of compared symbols from $A$ and $B$ 
(that is, only if they match or not), 
and its (at most) three already computed neighbor cells.
Each computed $M[i,j]$ cell, $1 \leq i \leq n, 1 \leq j \leq m$, 
stores the value of $LLCS(A[1\ldots i], B[1\ldots j])$.
A well-known property describes adjacent cells:
$M(i,j) - M(i-1,j) \in \{0, 1\}$ 
and $M(i,j) - M(i,j-1) \in \{0, 1\}$ for all valid $i$, $j$.

Despite almost 40 
years of research, surprisingly little can be said about 
the worst-case time complexity of LCS. 
It is known that in the very restrictive model of unconstrained alphabet 
and comparisons with equal/unequal answers only, the lower bound is 
$\Omega(mn)$ \cite{WC1976}, which is reached by a trivial DP algorithm.
%% If the input alphabet is fixed, the lower bound improves to 
%% $\Omega(\sigma n)$,
If the input alphabet is of constant size, 
the known lower bound is simply $\Omega(n)$, 
but if total order between alphabet symbols exists 
and $\leq$-comparisons are allowed, then the lower bound grows to 
$\Omega(n \log n)$ \cite{Hir1978}.
In other words, the gap between the proven 
lower bounds and the best worst-case algorithm is huge. 

A simple idea proposed in 1977 by Hunt and Szymanski \cite{HS1977} has become 
a milestone in LCS reseach, and the departure point for theoretically 
better algorithms (e.g.,~\cite{EGGI92}).
The Hunt--Szymanski (HS) algorithm is essentially based on dynamic programming, 
but it visits only the matching cells of the matrix, typically a small fraction 
of the entire set of cells. 
This kind of selective scan over the DP matrix is 
called {\em sparse dynamic programming} (SDP).
We note that the number of all matches in $M$, denoted with the symbol $r$, 
can be found in $O(n)$ time, 
and after this (negligible) preprocessing we can decide if 
the HS approach is promising to given data.
More precisely, the HS algorithm works in $O(n + r\log m)$ 
or even $O(n + r\log\log m)$ time.
Note that in the worst case, i.e., for $r = \Theta(mn)$, 
this complexity is however superquadratic.

The Hunt--Szymanski concept was an inspiration for a number of subsequent 
algorithms for LCS calculation, 
and the best of them, the algorithm of Eppstein et al.~\cite{EGGI92}, 
achieves $O(D \log\log(\min(D,mn/D)))$ worst-case time 
(plus $O(n\sigma)$ preprocessing), 
where $D \leq r$ is the number of so-called dominant matches in $M$ 
(a match $(i, j)$ is called dominant iff $M[i,j] = M[i-1,j]+1 = M[i,j-1]+1$).
Note that this complexity 
is $O(mn)$ for any value of $D$.
A more recent algorithm, by Sakai~\cite{S2012},
is an improvement if the alphabet is very small 
(in particular, constant), as its time complexity is 
$O(m\sigma + \min(D\sigma, p(m - q)) + n)$, 
where $p = LLCS(A, B)$ and 
$q = LLCS(A[1\ldots m], B)$.
%% The algorithms from the HS family are output-sensitive.

A different approach is to divide the dynamic matrix into small blocks, 
such that the number of essentially different blocks is small enough to be 
precomputed before the main processing phase.
In this way, the block may be processed in constant time each, making 
use of a built lookup table (LUT).
This ``Four Russians'' technique was first used to the LCS problem 
by Masek and Paterson~\cite{MP1980}, for a constant alphabet, 
and refined by Bille and Farach-Colton~\cite{BFC2008} to work with an 
arbitrary alphabet.
The obtained time compexities were 
$O(mn/\log^2 n)$ and $O(mn(\log\log n)^2/\log^2 n)$, respectively, 
with linear space.

A related, but different approach, is to use bit-paralellism to 
compute several cells of the dynamic programming matrix at a time. 
There are a few such variants (see~\cite{Hyy2004} and references therein),
all of them working in 
$O(\ceil{m/w}n)$ worst-case time, after $O(\sigma \ceil{m/w} + m)$-time 
and $O(\sigma m)$-space preprocessing, 
where $w \geq \log n$ is the machine word size.

Yet another line of research considers the input sequences in compressed 
form.
There exist such LCS algorithms for RLE-, LZ- 
and grammar-compressed inputs~\cite{LWL2008,CLZ2003,G2012}.
We briefly mention two results.
Crochemore et al.~\cite{CLZ2003} exploited the LZ78-factorization 
of the input sequences over a constant alphabet, 
to achieve $O(h mn /\log n)$ time, where $h \leq 1$ is the 
entropy of the inputs.
Gawrychowski~\cite{G2012} considered the case of two strings 
described by SLPs (straight line programs) of total size $n$, 
to show a solution computing their edit distance in 
$O(n N \sqrt{\log(N/n)})$ 
time, where $N$ is the sum of their (non-compressed) length.

Some other LCS-related results can be found in the surveys~\cite{Apo1997,BHR2000}.

%%%%%%%%%%%%%%%%%%%%%%%%%%%%%%%%%%%%%%%%%%%%%%%%%%%%%%%%%%%%%%%%%%%%%%%%%%%%%%
\section{LCS in $O(mn\log\log n/\log^2 n)$ time}
\label{sec:alg1}
%%%%%%%%%%%%%%%%%%%%%%%%%%%%%%%%%%%%%%%%%%%%%%%%%%%%%%%%%%%%%%%%%%%%%%%%%%%%%%

In this section we modify the technique of 
Bille and Farach-Colton (BFC)~\cite[Sect.~4]{BFC2008}, improving its 
worst-case time complexity by factor $\log\log n$, 
to achieve $O(mn\log\log n/\log^2 n)$ time, 
with linear space.

We divide the dynamic programming matrix $M[0 \ldots n, 0 \ldots m]$ 
into rectangular blocks with shared borders, of size $(x_1+1) \times (x_2+1)$, 
and process the matrix in horizontal stripes of $x_2$ rows.
By ``shared borders'' we mean that e.g. the bottom row of some block 
being part of its output is also part of the input of the block below.
Values inside each block depend on: 
\begin{description}
  \item[$(i)$] $x_1$ corresponding symbols from sequence $A$,
  \item[$(ii)$] $x_2$ corresponding symbols from sequence $B$,
  \item[$(iii)$] the top row of the block, which can be encoded 
                differentially in $x_1$ bits,
  \item[$(iv)$] the leftmost column of the block, which can be 
                 encoded differentially in $x_2$ bits.
\end{description}
We use the BFC technique of alphabet remapping in superblocks 
of size $y \times y$.
W.l.o.g. we assume that $x_2$ divides $y$.
For each substring $B[j'y+1 \ldots (j'+1)y]$
its symbols are sorted and 
%% $q = O(y)$ 
$q \leq y$
unique symbols are found.
Then, the $y$ symbols are remapped to $\Sigma_{B_{j'}} = \{0 \ldots q-1\}$, 
using a balanced BST.
Next, for each symbol from a snippet $A[i'y+1 \ldots (i'+1)y]$ 
we find its encoding in $\Sigma_{B_{j'}}$, or assign $q$ to it 
if it wasn't found there.
This takes $O(\log y)$ time per symbol, 
and the overall alphabet remapping time for the whole matrix 
is $O(m\log y + mn\log y/y)$.

This remapping technique allows to represent the symbols from 
the input components $(i)$ and $(ii)$ on 
$O(\log\min(y+1, \sigma))$ bits each, 
rather than $\Theta(\log\sigma)$ bits.
It works because not the actual symbols from $A$ and $B$ 
are important for LCS computations, but only 
equality 
relations between them.
To simplify notation, let us assume a large enough alphabet 
so that $\min(y+1, \sigma) = y+1$.

Now, for each (remapped) substring of length $x_2$ from sequence $B$ 
we build a lookup table for fast handling of the blocks in one horizontal stripe.
Once a stripe is processed, its LUT is discarded to save space.
This requires to compute the answers for all possible inputs in components 
$(i)$, $(iii)$ and $(iv)$ (the component $(ii)$ is fixed for a given stripe).
The input thus takes 
$x_1 \log(y+1) + x_1 + x_2 = x_1 \log(2(y+1)) + x_2$ bits.

The return value associated with each LUT key are the bottom and 
the right border of a block, in differential form 
(the lowest cell in the right border and the rightmost cell in the bottom 
border are the same cell, which is represented twice; 
once as a difference (0 or 1) to its left neighbor in the bottom border 
and once as a difference (0 or 1) to its upper neighbor in the right border)
and the difference between the values of the bottom right and the top left 
corner (to know the explicit value of $M$ in the bottom right corner), 
requiring $x_1 + x_2 + \log(\min(x_1, x_2)+1) $ bits in total.
Fig.~\ref{fig:dp} illustrates.

\begin{figure}[pt]
\centerline{
\includegraphics[width=0.98\textwidth,scale=1.0]{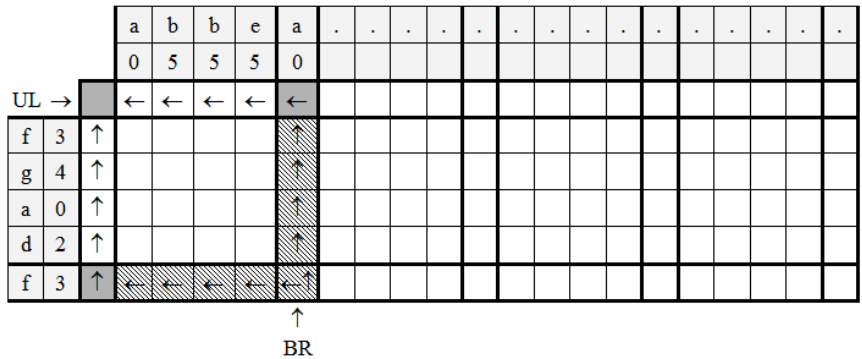}
}
\caption{One horizontal stripe of the DP matrix, with 4 blocks of size 
$5 \times 5$ ($x_1 = x_2 = 4$).
The corresponding snippets from sequence $A$ and $B$ are \texttt{abbea} 
and \texttt{fgadf}, respectively.
These snippets are translated to a new alphabet (the procedure for creating 
the new alphabet is not shown here) of size 6, where the characters from $A$ 
are mapped onto the alphabet $\{0, 1, \ldots, 4\}$ and value $5$ is used for 
the characters from $B$ not used in the encoding of the symbols from 
$A$ belonging to the current superblock (the superblock is not shown here).
The LCS values are stored explictly in the dark shaded cells.
The white and dark shaded cells with arrows are part of the input, and their 
LCS values are encoded differentially, with regard to their left or upper neighbor.
The diagonally shaded cells are the output cells, also encoded differentially.
The bottom right corner (BR) is stored in three forms: 
as the difference to its left neighbor (0 or 1), 
as the difference to its upper neighbor (0 or 1) 
and the value of UL plus the difference between BR and UL.
The difference between BR and UL is part of the LUT output for the current block.}
\label{fig:dp}
\end{figure}

As long as the input and the output of a LUT fits a machine word, i.e., 
does not exceed $w$ bits, where $w \geq \log n$, we will process 
one block in constant time. 
Still, the LUT building costs also impose a limitation.
More precisely, we are going to minimize the total time of remapping 
the alphabet in all the superblocks, building all $O(m/x_2)$ 
LUTs and finally processing all the blocks, which is described by 
the formula:
\begin{equation*}
O(m\log y + mn\log y/y +
   (m/x_2) 2^{x_1 \log(2(y+1)) + x_2} x_1 x_2 + 
   mn/(x_1 x_2)), 
\end{equation*}
where $2^{x_1 \log(2(y+1)) + x_2}$ is the number of all possible LUT 
inputs and the $x_1 x_2$ multiplier corresponds to the computation time 
per one LUT cell.
Let us set $y = \log^2 n / 2$, 
$x_1 = \log n / (4\log\log n)$ and $x_2 = \log n / 4$.
In total we obtain $O(mn\log\log n /\log^2 n)$ time 
with $o(n)$ extra space (for the lookup tables, used one at a time, 
and alphabet remapping), 
which improves the Bille and Farach-Colton result by factor $\log\log n$.
The improvement is achieved thanks to using multiple lookup tables 
(one per horizontal stripe).
Formally, we obtain the following theorem.

\begin{theorem}\label{th:lcs1}
The length of the longest common subsequnce (LCS) between two sequences, 
$A$, of length $n$, and $B$, of length $m$, 
where $n \geq m \geq \log^2 n$, both over an integer alphabet, 
can be computed in $O(mn\log\log n /\log^2 n)$ worst-case time.
The algorithm needs $o(n)$ words of space, apart for the two sequences 
themselves.
\end{theorem}

%%%%%%%%%%%%%%%%%%%%%%%%%%%%%%%%%%%%%%%%%%%%%%%%%%%%%%%%%%%%%%%%%%%%%%%%%%%%%%
\section{LCS in $O(mn/\log^2 n + r)$ time (for some $r$)}
\label{sec:alg2}
%%%%%%%%%%%%%%%%%%%%%%%%%%%%%%%%%%%%%%%%%%%%%%%%%%%%%%%%%%%%%%%%%%%%%%%%%%%%%%

In this algorithm we also work in blocks, of size $(b+1) \times (b+1)$, 
but divide them into two groups: 
sparse blocks are those which contain at most $K$ matches 
and dense blocks are those which contain more than $K$ matches.
Obviously, we do not count possible matches on the input boundaries of 
a block.

We observe that knowing the left and top boundary of a block 
plus the location of all the matches in the block is enough to 
compute the remaining (right and bottom) boundaries.
This is a nice property as it eliminates the need to (explicitly) 
access the corresponding substrings of $A$ and $B$.

The sparse block input will be encoded as:
\begin{description}
  \item[$(i)$] the top row of the block, represented 
               differentially in $b$ bits,
  \item[$(ii)$] the leftmost column of the block, represented 
               differentially in $b$ bits,
  \item[$(iii)$] the match locations inside the block, 
                each in $\log(b^2)$ bits, totalling $O(K \log b)$ bits.
\end{description}

Each sparse block will be computed in constant time, thanks to a LUT.
Dense blocks, on the other hand, will be partitioned into smaller blocks, 
which in turn will be handled with our algorithm from 
Section~\ref{sec:alg1}.
Clearly, we have $b = O(\log n)$ (otherwise the LUT build costs 
would be dominating)
and $b = \omega(\log n/\sqrt{\log\log n})$
(otherwise this algorithm would never 
be better than the one from Section~\ref{sec:alg1}),
which implies that 
$K = \Theta(\log n / \log\log n)$, with an appropriate constant.

As this algorithm's worst-case time is $\Omega(mn/\log^2 n)$, 
it is easy to notice that the preprocessing costs 
for building required LUTs and alphabet mapping will not dominate. 
Each dense block is divided into smaller blocks of size 
$\Theta(\log n/\log\log n) \times \Theta(b)$.
Let the fraction of dense blocks in the matrix be denoted as $f_d$. 
The total time complexity (without preprocessing) is then
\begin{equation*}
O((1-f_d) mn/b^2 + f_d (mn\log\log n/(b\log n))).
\end{equation*}
The fraction $f_d$ must be $o(1)$, otherwise this algorithm 
is not better in complexity than the previous one.
This also means that $1-f_d$ may be replaced with $1$ in further 
complexity considerations.

Recall that $r$ is the number of matches in the DP matrix.
We have $f_d = O((r/K) / (mn/b^2)) = O(rb^2\log\log n/(mn\log n))$.
From the $f_d = o(1)$ condition we also obtain that 
$rb^2 = o(mn\log n / \log\log n)$.
If $r = o(mn / (\log n \log\log n))$, 
then we can safely use 
the maximum possible value of $b$, i.e., $b = \Theta(\log n)$ 
and obtain the 
time of $O(mn/\log^2 n)$.

Unfortunately, in the preprocessing we have to find  
and encode all matches in all sparse blocks, which requires $O(n + r)$ time.
Overall, this leads to the following theorem.

\begin{theorem}\label{th:lcs2}
The length of the longest common subsequnce (LCS) between two sequences, 
$A$, of length $n$, and $B$, of length $m$, 
where $n \geq m \geq \log^2 n$, both over an integer alphabet, 
can be computed in $O(mn/\log^2 n + r)$ worst-case time, 
assuming $r = o(mn / (\log n \log\log n))$, 
where $r$ is the number of matching pairs of symbols between $A$ and $B$.
\end{theorem}

%% This gives the overall $O(mn/\log^2 n + r)$ time 
%% when $r = o(mn / (\log n \log\log n))$, 
%% and considering this restriction on $r$ this is better 
Considering to the presented restriction on $r$, 
the achieved complexity is better than 
than the result from the previous section. 

On the other hand, it is essential to compare the obtained time complexity 
with the one from Eppstein et al. algorithm~\cite{EGGI92}.
All we know about the number of dominant matches $D$ is that 
$D \leq r$\footnote{A slightly more precise bound on $D$ is 
$\min(r, m^2)$, but it may matter, in complexity terms, 
only if $m = o(n)$ (cf. also~\cite[Th.~1]{S2012}), 
which is a less interesting case.},
so we replace $D$ with $r$ in their complexity formula
to obtain $O(r \log\log(\min(r,mn/r)))$ in the worst case. 
Our result is better if $r = \omega(mn/(\log^2 n \log\log\log n))$ 
and $r = o(mn)$.
Overall, it gives the niche of 
$r = \omega(mn/(\log^2 n \log\log\log n))$ 
and 
$r = o(mn\log\log n/\log^2 n)$
in which the algorithm presented in this section is competitive.

The alphabet size is yet another constraint.
From the comparison to Sakai's algorithm~\cite{S2012} 
we conclude that our algorithm needs $\sigma = \omega(\log\log\log n)$ 
to dominate for the case of 
$r = \omega(mn/(\log^2 n \log\log\log n))$.

%%%%%%%%%%%%%%%%%%%%%%%%%%%%%%%%%%%%%%%%%%%%%%%%%%%%%%%%%%%%%%%%%%%%%%%%%%%%%%
\section{Algorithmic applications}
\label{sec:apps}
%%%%%%%%%%%%%%%%%%%%%%%%%%%%%%%%%%%%%%%%%%%%%%%%%%%%%%%%%%%%%%%%%%%%%%%%%%%%%%

The techniques presented in the two previous sections may be applied 
to any sequence similarity problem fulfilling certain properties.
The conditions are specified in the following lemma.

\begin{lemma}\label{lem:2d}
Let $Q$ be a sequence similarity problem 
returning the length of a desired subsequence, 
involving two sequences, $A$ of length $n$ and $B$ of length $m$, 
both over a common integer alphabet $\Sigma$ of size $\sigma = O(m)$.
We assume that $1 \leq m \leq n$.
Let $Q$ admit a dynamic programming solution in which 
$M(i,j) - M(i-1,j) \in \{-1, 0, 1\}$, 
$M(i,j) - M(i,j-1) \in \{-1, 0, 1\}$
for all valid $i$ and $j$, 
and $M(i,j)$ depends only on the values of its (at most) three 
neighbors $M(i-1,j)$, $M(i,j-1)$, $M(i-1,j-1)$, 
and whether $A_i = B_j$.

There exists a solution to $Q$ with $O(mn\log\log n/\log^2 n)$
worst-case time.
There also exists a solution to $Q$ with $O(mn/\log^2 n + r)$ 
worst-case time, for 
$r = o(mn / (\log n \log\log n))$, 
where $r$ is the number of symbols pairs 
$A_i$, $B_j$ such that $A_i = B_j$.
The space use in both solutions is $O(n)$ words.
\end{lemma}

\begin{proof}
We straightforwardly apply the ideas presented in the previous 
two sections.
The only modification is to allow a broader range of 
differences ($\{-1, 0, 1\}$) between adjacent cells in the 
dynamic programming matrix. 
This only affects a constant factor in parameter setting.
\qed
\end{proof}

Lemma~\ref{lem:2d} immediately serves to calculate the edit (Levenshtein) 
distance between two sequences (in fact, the BFC technique 
was presented in terms of the edit distance).
We therefore obtain the following theorem.

\begin{theorem}\label{th:edit}
The edit distance between two sequences, $A$, of length $n$, 
and $B$, of length $m$, where $n \geq m \geq \log^2 n$, 
both over an integer alphabet, 
can be computed in $O(mn\log\log n /\log^2 n)$ worst-case time.
Alternatively, the distance can be found in $O(mn/\log^2 n + r)$ 
worst-case time, for $r = o(mn / (\log n \log\log n))$, 
where $r$ is the number of symbols pairs 
$A_i$, $B_j$ such that $A_i = B_j$.
The space use in both solutions is $O(n)$ words.
\end{theorem}

Another feasible problem is the longest common transposition-invariant 
subsequence (LCTS)~\cite{MNU2004,D2006}, in which 
%% ~\cite{LU2000,MNU2004,D2006}
we look for a longest subsequence of the form 
$(s_1 + t)(s_2 + t) \ldots (s_{\ell} + t)$ 
such that all $s_i$ belong to $A$ (in increasing order), 
all corresponding values $s_i + t$ belong to $B$ (in increasing order), 
and $t \in \{-\sigma+1 \ldots \sigma-1\}$ is some integer, called a transposition.
This problem is motivated by music information retrieval.
The best known results for LCTS are 
$O(mn\log\log\sigma)$~\cite{NGMD2005,D2006} and 
$O(mn\sigma(\log\log n)^2/\log^2 n)$ if the BFC technique is applied 
for all transpositions (which is $O(mn)$ if $\sigma = O(\log^2 n/(\log\log n)^2)$).
Applying the former result from Lemma~\ref{lem:2d}, 
for all possible transpositions, gives immediately 
$O(mn\sigma\log\log n/\log^2 n)$ time complexity 
(if $\sigma = O(n^{1-\varepsilon})$, for any $\varepsilon > 0$, 
otherwise the LUT build costs would dominate).
Applying the latter result requires more care.
First we notice that the number of matches over 
all the transpositions sum up to $mn$, 
so $\Theta(mn)$ is the total preprocessing cost.
Let us divide the transpositions into dense ones 
and sparse ones, where the dense ones are those that have 
at least $mn\log\log n/\sigma$ matches.
The number of dense transpositions is thus limited to 
$O(\sigma/\log\log n)$.
We handle dense transpositions with the technique from Section~\ref{sec:alg1} 
and sparse ones with the technique from Section~\ref{sec:alg2}.
This gives us 
$O(mn + mn(\sigma/\log\log n)\log\log n/\log^2 n + 
mn\sigma/\log^2 n) = O(mn(1 + \sigma/\log^2 n))$ total time, 
assuming that 
$\sigma = \omega(\log n (\log\log n)^2)$, 
as this condition on $\sigma$ implies the number of matches 
in each sparse transposition limited to 
$o(mn / (\log n \log\log n))$, 
as required.
We note that $\sigma = \omega(\log^2 n/(\log\log n)^2)$ 
and $\sigma = O(\log^2 n)$ is the niche in which our algorithm 
is the first one to achieve $O(mn)$ total time.

\begin{theorem}\label{th:lcts}
The length of the longest common transposition-invariant 
subsequence (LCTS) 
between two sequences, $A$, of length $n$, 
and $B$, of length $m$, where $n \geq m \geq \log^2 n$, 
both over an integer alphabet of size $\sigma$, 
can be computed in $O(mn(1 + \sigma/\log^2 n))$ worst-case time, 
assuming that $\sigma = \omega(\log n (\log\log n)^2)$.
\end{theorem}

A natural extension of Lemma~\ref{lem:2d} is to involve more than two 
(yet a constant number of) sequences. 
In particular, problems on three sequences have practical importance.

\begin{lemma}\label{lem:3d}
Let $Q$ be a sequence similarity problem 
returning the length of a desired subsequence, 
involving 
three sequences, $A$ of length $n$, $B$ of length $m$ and $P$ of length $u$, 
all over a common integer alphabet $\Sigma$ of size $\sigma = O(m)$.
We assume that $1 \leq m \leq n$ 
and $u = \Omega(n^c)$, for some constant $c > 0$.
Let $Q$ admit a dynamic programming solution in which 
$M(i,j,k) - M(i-1,j,k) \in \{-1, 0, 1\}$, 
$M(i,j,k) - M(i,j-1,k) \in \{-1, 0, 1\}$
and
$M(i,j,k) - M(i,j,k-1) \in \{-1, 0, 1\}$, 
for all valid $i$, $j$ and $k$,
and $M(i,j)$ depends only on the values of its 
(at most) seven neighbors:
$M(i-1,j,k)$, $M(i,j-1,k)$, $M(i-1,j-1,k)$, 
$M(i,j,k-1)$, $M(i-1,j,k-1)$, $M(i,j-1,k-1)$ and $M(i-1,j-1,k-1)$, 
and whether $A_i = B_j$, $A_i = P_k$ and $B_j = P_k$.

There exists a solution to $Q$ with $O(mnu/\log^{3/2} n)$
worst-case time.
The space use is $O(n)$ words.
\end{lemma}

\begin{proof}
The solution works on cubes of size $b \times b \times b$, 
setting $b = \Theta(\sqrt{\log n})$ with an appropriate constant.
Instead of horizontal stripes, 3D ``columns'' of size $b \times b \times u$ 
are now used.
The LUT input consists of $b$ symbols from sequence $P$, 
encoded with respect to a supercube in $O(\log\log n)$ bits each, 
and three walls, of size $b \times b$ each, 
in differential representation.
The output are the three opposite walls of a cube.
The restriction $u = \Omega(n^c)$ implies that 
the overall time formula {\em without the LUT build times} is 
$\Omega(mn^{1+c}/\log^{3/2} n)$, which is $\Omega(mn^{1+c'})$,
for some constant $c'$, $c \geq c' > 0$.
The build time for all LUTs can be made $O(mn^{1+c''})$, 
for any constant $c'' > 0$, if the constant associated with $b$ 
is chosen appropriately.
We now set $c'' = c'$ to show the build time for the LUTs 
is not dominating.
\qed
\end{proof}

As an application of Lemma~\ref{lem:3d} we present 
the merged longest common subsequence (MerLCS) problem~\cite{HYTAP2008}, 
which involves three sequences, $A$, $B$ and $P$, 
and its returned value is a longest sequence $T$ that is 
a subsequence of $P$ and can be split into two subsequences 
$T'$ and $T''$ such that $T'$ is a subsequence of $A$ 
and $T''$ is a subsequence of $B$.
Deorowicz and Danek~\cite{DD2013} showed that in the DP formula 
for this problem $M(i,j,k)$ is equal to or larger by 1 than 
any of the neighbors: 
$M(i-1,j,k)$,
$M(i,j-1,k)$
and $M(i,j,k-1)$.
They also gave an algorithm working 
in $O(\lceil u/w \rceil mn \log w)$ time.
Peng et al.~\cite{PYHTH2010} gave an algorithm
with $O(\ell mn)$ time complexity, where $\ell \leq n$ is the 
length of the result.
Motivations for the MerLCS problem, from bioinformatics 
and signal processing, can be found e.g. in~\cite{DD2013}.

Based on the cited DP formula property~\cite{DD2013} 
we can apply Lemma~\ref{lem:3d} to 
obtain $O(mnu/\log^{3/2} n)$ 
time for MerLCS (if $u = \Omega(n^c)$ for some $c > 0$), 
which may be competitive with existing solutions.

\begin{theorem}\label{th:merlcs}
The length of the merged longest common subsequence (MerLCS)
involving three sequences, $A$, $B$ and $P$, of length respectively 
$n$, $m$ and $u$, where $m \leq n$ 
and $u = \Omega(n^c)$, for some constant $c > 0$, 
all over an integer alphabet of size $\sigma$, 
can be computed in $O(mnu/\log^{3/2} n)$ worst-case time.
\end{theorem}

%%%%%%%%%%%%%%%%%%%%%%%%%%%%%%%%%%%%%%%%%%%%%%%%%%%%%%%%%%%%%%%%%%%%%%%%%%%%%%
\section{Conclusions}
\label{sec:conc}
%%%%%%%%%%%%%%%%%%%%%%%%%%%%%%%%%%%%%%%%%%%%%%%%%%%%%%%%%%%%%%%%%%%%%%%%%%%%%%

On the example of the longest common subsequence problem 
we presented two algorithmic techniques, making use of tabulation 
and sparse dynamic programming paradigms, 
which allow to obtain competitive time complexities.
Then we generalize the ideas by specifying conditions on 
DP dependencies whose fulfilments lead to immediate 
applications of these techniques.
The actual problems considered here as applications 
comprise the edit distance, LCTS and MerLCS.

As a future work, we are going to relax the DP dependencies, 
which may for example improve the SEQ-EC-LCS result from~\cite{DG2014}.
Another research option is to try to improve the tabulation based 
result on compressible sequences.

\section*{Acknowledgments}
The author wishes to thank Sebastian Deorowicz 
for helpful comments on a preliminary version of the manuscript.

\bibliographystyle{abbrv}
\bibliography{lcs}

\end{document}